\documentclass{stacs_proc}

\usepackage{graphics} 

\def\G{\Gamma}         
\newcommand{\code}[1]{\textsc{#1}}
\newcommand{\U}{\mathcal{U}}
\newcommand{\OPT}{\mathit{OPT}}
\newcommand{\tempone}{u-red}

\begin{document}
\title{Computing Minimum Spanning Trees with Uncertainty}
\author[lab1]{T. Erlebach}{Thomas Erlebach}

\address[lab1]{Department of Computer Science, University of Leicester, UK.}
\email{{te17,mh55,rr29}@mcs.le.ac.uk}


\author[lab1]{M. Hoffmann}{Michael Hoffmann}

\author[lab2]{D. Krizanc}{Danny Krizanc}

\address[lab2]{Department of Mathematics and Computer Science, Wesleyan University, USA.}
\email{dkrizanc@wesleyan.edu}

\author[lab3]{M. Mihal'\'{a}k}{Mat\'{u}\v{s} Mihal'\'{a}k}

\address[lab3]{Institut f\"{u}r Theoretische Informatik, ETH
  Z\"{u}rich, Switzerland.}
\email{matus.mihalak@inf.ethz.ch}

\author[lab1]{R. Raman}{Rajeev Raman}


\thanks{Part of the presented research was undertaken while the second
  and fifth author were on study leave from the University of
  Leicester and the second author was visiting the University of
  Queensland. The authors would like to thank these institutions for
  their support.}

\begin{abstract}
We consider the minimum spanning tree problem in a setting
where information about the edge weights of the given graph
is uncertain. Initially, for each edge $e$ of the
graph only a set $A_e$, called an \emph{uncertainty area}, that
contains the actual edge weight $w_e$ is known. The algorithm
can `update' $e$ to obtain the edge weight $w_e \in A_e$.
The task is to output the edge set of a minimum spanning
tree after a minimum number of updates. An algorithm
is $k$-update competitive if it makes at most $k$
times as many updates as the optimum.  
We present a $2$-update competitive algorithm if all 
areas $A_e$ are open or trivial, which is the
best possible among deterministic algorithms. The 
condition on the areas $A_e$ is to exclude degenerate
inputs for which no constant update
competitive algorithm can exist. 

Next, we consider a setting where the vertices of the graph correspond
to points in Euclidean space
and the weight of an edge is equal to the distance
of its endpoints. The location of each point is initially
given as an uncertainty area,
and an update reveals the exact location of the point.
We give a general relation between the edge uncertainty
and the vertex uncertainty versions of a problem
and use it to derive a $4$-update competitive algorithm
for the minimum spanning tree problem in the vertex
uncertainty model. Again, we show that this is best
possible among deterministic algorithms.
\end{abstract}

\maketitle

\stacsheading{2008}{277-288}{Bordeaux}
\firstpageno{277}

\section{Introduction}
In many applications one has to deal with computational
problems where some parts of the input data are imprecise
or uncertain. For example, in a geometric problem involving
sets of points, the locations of the points might be known
only approximately; effectively this means that instead
of the location of a point, only a region or area containing
that point is known. In other applications, only estimates
of certain input parameters may be known, for example in
form of a probability distribution. There are many different
approaches to dealing with problems of this type, including
e.g.\ stochastic optimization and robust optimization.

Pursuing a different approach, we consider a setting in which
the algorithm can obtain exact information about an input data item
using an \emph{update} operation, and we are interested
in the \emph{update complexity} of an algorithm, i.e., our goal
is to compute a correct solution using a minimum number of
updates. The updates are adaptive, i.e., one selects the
next item to update based on the result of the updates
performed so far, so we refer to the algorithm as an \emph{on-line
algorithm}. There are a number of application areas where
this setting is meaningful. For example, in a mobile
ad-hoc network an algorithm may have knowledge about
the approximate locations of all nodes, and it is
possible (but expensive) to find out the exact current
location of a node by communicating to that node and
requesting that information. To assess the performance
of an algorithm, we compare the number of updates
that the algorithm makes to the optimal number of updates.
Here, optimality is defined in terms of an
adversary, who, knowing the values of
all input parameters, makes the fewest updates
needed to present a solution to the problem that is
provably correct, in that no additional areas need
to be updated to verify the correctness of the solution
claimed by the adversary.
We say that an algorithm is \emph{$k$-update competitive}
if, for each input instance, the algorithm makes at most $k$ times as
many updates as the optimum number of updates for that input instance.
The notions of update complexity and $k$-update
competitive algorithms were implicit
in Kahan's model for data in motion~\cite{Kahan/91}
and studied further for two-dimensional geometric
problems by Bruce et al.\ \cite{Bruce_et_al/05}.

In this paper, we consider
the classical minimum spanning tree (MST) problem in two
settings with uncertain information. In the first setting,
the edge weights are initially given as uncertainty
areas, and the algorithm can obtain the exact weight
of an edge by updating the edge. If the uncertainty
areas are trivial (i.e., contain a single number)
or (topologically) open, we give a $2$-update competitive algorithm
and show that this is best possible for deterministic
algorithms.  Without this restriction on the areas,
it is easy to construct degenerate inputs for which 
there is no constant update competitive algorithm.
Although degeneracy could also be excluded by other
means (similar to the ``general position'' assumption
in computational geometry), our condition is much cleaner.

In the second setting that
we consider, the vertices of the graph correspond
to points in Euclidean space, and the locations of the points are
initially given as uncertainty areas. The weight
of an edge equals the distance between the
points corresponding to its vertices.
The algorithm can update a vertex to reveal its
exact location. We give a general relation between
the edge uncertainty version and the vertex uncertainty
version of a problem. For trivial or open uncertainty
areas we obtain a $4$-update competitive algorithm
for the MST problem with vertex
uncertainty and show again that this is optimal for
deterministic algorithms.

\smallskip
\noindent
\textbf{Related Work.\ }
We do not attempt to survey the vast literature dealing
with problems on uncertain data, but focus on work most
closely related to ours. Kahan \cite{Kahan/91} studied the
problem of finding the maximum, the median and the minimum gap
of a set of $n$ real values constrained to fall in
a given set of $n$ real intervals. In the spirit of
competitive analysis, he defined the \emph{lucky ratio}
of an update strategy as the worst-case ratio between the
number of updates made by the strategy and the optimal
number of updates of a non-deterministic strategy.
In our terminology, a strategy with lucky ratio $k$
is $k$-update competitive. Kahan gave strategies with
optimal lucky ratios for the problems considered \cite{Kahan/91}.

Bruce et al.\ studied the problems of computing maximal points or the
points on the convex hull of a set of uncertain points
\cite{Bruce_et_al/05} and presented $3$-update competitive
algorithms. They introduced a general method, called
the \emph{witness algorithm}, for dealing with problems
involving uncertain data, and derived their $3$-update
competitive algorithms using that method.
The algorithms we present in this paper are based on the
method of the witness algorithm of \cite{Bruce_et_al/05},
but the application to
the MST problem is non-trivial.

Feder et al.\ \cite{Feder_median/03,Feder_shortest/07},
consider two problems in a similar framework to ours.
Firstly, they consider the problem of computing the median
of $n$ numbers to within a given tolerance.
Each input number lies in an interval, and an update reveals
the exact value, but different intervals have different update
costs.  They consider off-line algorithms, which must decide
the sequence of updates prior to seeing the answers, as well
as on-line ones, aiming to minimize the total update cost. 
In \cite{Feder_shortest/07}, off-line algorithms
for computing the length of a shortest path from a 
source $s$ to a given vertex $t$ are considered.  Again,
the edge lengths lie in intervals with different update costs, and 
they study the computational complexity of
minimizing the total update cost.

One difference between the framework of Feder et al.\ and ours is
that they require the computation of a specific numeric value
(the value of the median, the length of a shortest path).  We,
on the other hand, aim to obtain a subset of edges that form
an MST.  In general, our version of the problem may require far 
fewer updates.  Indeed, for the MST with vertex uncertainties, 
it is obvious that one must update all non-trivial areas
to compute the cost of the MST exactly.  However, the cost of the
MST may not be needed in many cases: if the MST is to be used
as a routing structure in a wireless ad-hoc network, then it
suffices to determine the edge set.  Also, our algorithms
aim towards on-line optimality against an adversary, whereas their 
off-line algorithms aim for static optimality.

Further work in this vein attempts to compute other aggregate functions to 
a given degree of tolerance, and establishes tradeoffs between update costs
and error tolerance or presents complexity results for
computing optimal strategies, see e.g.\ \cite{Olston-Widom/00, Khanna-Tan/01}.

Another line of work considers the robust spanning tree
problem with interval data.
For a given graph with weight
intervals specified for its edges, the goal is to compute a spanning
tree that minimizes the worst-case deviation from the minimum spanning
tree (also called the \emph{regret}), over all realizations of the edge
weights. This is an off-line problem, and no update operations
are involved. The problem is proved $\mathcal{NP}$-hard in~\cite{Aron_et_al/04}.
A $2$-approximation algorithm is given in~\cite{Kasperski-Zielinski/05}.
Further work has considered heuristics or exact algorithms for the problem,
see e.g.\ \cite{Yaman_et_al/01}.

In the setting of geometric problems with imprecise points,
L\"offler and van Kreveld have studied the problem of
computing the largest or smallest convex hull over all
possible locations of the points inside their uncertainty
areas~\cite{Loeffler-vanKreveld/06}. Here, the option
of updating a point does not exist, and the goal is
to design fast algorithms computing an extremal solution
over all possible choices of exact values of the input data.

\smallskip
The remainder of the paper is organized as follows.
In Section~\ref{sec:prelim}, we define our problems and introduce
the witness algorithm of \cite{Bruce_et_al/05} in general form. 
Sections~\ref{sec:edgeu}
and~\ref{sec:vertexu}
give our results for MSTs with edge 
and vertex uncertainty, respectively.

\section{Preliminaries}
\label{sec:prelim}

The \code{mst-edge-uncertainty} problem is defined as follows:
Let $G=(V,E)$ be a connected, undirected, weighted graph. Initially the edge weights
$w_e$ are unknown; instead, for each edge $e$ an area
$A_e$ is given with $w_e \in A_e$. 
When updating an edge $e$, the value of $w_e$ is
revealed. The aim is to find (the edge set of) an MST
for $G$ with the least number of updates.

In applications such as mobile ad-hoc networks it 
is natural to assume the vertices of our graph are embedded in two
or three dimensional space.  This leads
to the \code{mst-vertex-uncertainty} problem defined as follows:
Let $G=(V,E)$ be a connected, undirected, weighted graph. The vertices
correspond to points in Euclidean space.
We refer to the point $p_v$ corresponding to a vertex $v$ as its
\emph{location}. The weight of an edge is the
Euclidean distance  between the locations of its vertices.
Initially the locations of the vertices are not known; instead,
for each vertex $v$ an area
$A_v$ is given with $p_v \in A_v$, where $p_v$ is the actual location of
vertex~$v$.
When a vertex $v$ is updated, the location $p_v$ is
revealed. The aim is to find an MST for $G$ with
the least number of updates.

Formally we are interested in on-line update problems of the following type:
Each problem instance $P=(C,A,\phi)$ consists of an ordered set of data
$C=\{c_1,\dots,c_n\}$, also called a configuration,
and a function $\phi$ such that $\phi(C)$ is the
set of solutions for $P$. (The function $\phi$ is the
same for all instances of a problem and can thus be taken
to represent the problem.)
At the beginning the set $C$ is not known to the algorithm; instead,
an ordered set of areas $A=\{A_1,\dots, A_n\}$ is given, such
that $c_i\in C$ is an element of $A_i$.
The sets $A_i$ are called \emph{areas of
uncertainty} or \emph{uncertainty areas} for $C$.
We say that an uncertainty area $A_i$ that consists of a single element
is \emph{trivial}.
For example, in the \code{mst-edge-uncertainty} problem, $C$ consists of
the given graph $G=(V,E)$ and its $|E|$ actual edge weights.
The ordered set of areas $A$ specifies the graph $G$ exactly (so we assume
complete knowledge of $G$) and, for each
edge~$e\in E$, contains an area $A_e$ giving the possible
values the weight of $e$ may take. Then $\phi(C)$ is the set of
MSTs of the graph with edge weights given by $C$,
each tree represented as a set of edges.

For a given set of uncertainty areas $A=\{A_1,\dots, A_n\}$,
an area $A_i$ can be \emph{updated}, which
reveals the exact value of $c_i$. 
After updating $A_i$, the new ordered set of areas of
uncertainty for $C$ is
$\{A_1,\dots,A_{i-1},\{c_i\},A_{i+1}, \dots, A_n \}$.
Updating all non-trivial areas would reveal the configuration $C$
and would obviously allow us to calculate an element of $\phi(C)$
(under the natural assumption that $\phi$ is computable). 
The aim of the on-line algorithm is to
minimize the number of updates needed in order to compute an element
of $\phi(C)$. 

An algorithm is $k$-update competitive for a given problem
$\phi$ if for every
problem instance $P=(C,A,\phi)$ the algorithm
needs at most $k\cdot \OPT+c$ updates,
where $c$ is a constant and $\OPT$ is the minimum
number of updates needed to verify an element of $\phi(C)$.
(For our algorithms we can take $c=0$, but our lower bounds
apply also to the case where $c$ can be an arbitrary constant.)
Note that the primary aim is to minimize the number of updates needed
to calculate a solution. We do not consider running time or space
requirements in detail, but note that our algorithms are clearly
polynomial, provided that one can obtain the infimum and supremum
of an area in $O(1)$ time, an assumption which holds e.g.\ if areas are
open intervals.

\begin{figure}
  \centering
  \scalebox{.4}{\includegraphics{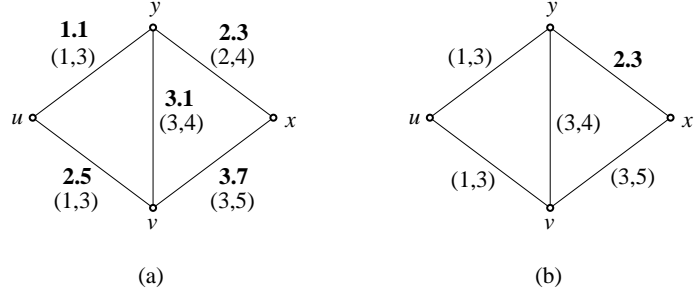}}
  \caption{(a) Instance of \code{mst-edge-uncertainty} $\;$
    (b) Updating the edge
    $\{x,y\}$ suffices to verify an MST}
  \label{fig:example}
\end{figure}%
As an example, consider the instance of
\code{mst-edge-uncertainty} shown in Figure~\ref{fig:example}(a),
where each edge is labeled with its actual weight (in bold) and its
uncertainty area (an open interval). Updating the edge $\{x,y\}$
leads to the situation shown in Figure~\ref{fig:example}(b) and
suffices to verify that the edges $\{u,y\}$, $\{u,v\}$ and $\{x,y\}$
form an MST regardless of the exact weights of
the edges that have not yet been updated. If no edge is updated, one cannot
exclude that an MST includes the edge $\{v,x\}$
instead of $\{x,y\}$, as the former could have weight $3.3$ and
the latter weight $3.9$, for example. Therefore, for the instance
of \code{mst-edge-uncertainty} in Figure~\ref{fig:example}(a)
the minimum number of updates is~$1$.

\subsection{The Witness Algorithm}
\label{sec:witness}%
The witness algorithm for problems with uncertain input
was first introduced in \cite{Bruce_et_al/05}. 
This section describes the witness algorithm in a more
general setting and notes some of its properties.
We call $W \subseteq A$ a {\it witness set} of
$(A,\phi)$ if for every possible configuration $C$ (where
$c_i\in A_i$) no element of
$\phi(C)$ can be verified without updating an element of $W$.
In other words, any set of updates that suffices to verify
a solution must update at least one area of~$W$.
The witness algorithm for a problem instance $P=(C,A,\phi)$ is
shown in Figure~\ref{fig:witnessalgo}.%
\begin{figure}
\centering\begin{minipage}{23.4cm}\begin{tabbing}
ww\=ww\=ww\=ww\=\kill
  \textbf{if} an element of $\phi(C)$ can not be calculated from $A$ \textbf{then}\+\\
    find a witness set $W$\\
    update all areas in $W$\\
    let $A'$ be the areas of uncertainty after updating $W$\\
    restart the algorithm with $P'=(C,A',\phi)$\-\\
  \textbf{end if}\\
  \textbf{return} an element of $\phi(C)$ that can be calculated from $A$
\end{tabbing}\end{minipage}
\caption{The general witness algorithm}
\label{fig:witnessalgo}
\end{figure}

For two ordered sets of areas $A=\{A_1,A_2,\dots,A_n\}$ and
$B=\{B_1,B_2,\dots,B_n\}$ we say that $B$ is at least as 
narrow as $A$ if $B_i \subseteq A_i$ for all $1\le i\le n$. 
The following lemma is easy to prove.

\begin{lemma}\label{lem:narrow}
Let $P=(C,A,\phi)$ be a problem instance and $B$ be a narrower set of areas
than $A$. 
Further let $W$ be a witness set of $(B,\phi)$. Then $W$ is also a
witness set of $(A,\phi)$.
\end{lemma}

\begin{theorem}
\label{th:global_bound}
If there is a global bound $k$ on the size of any witness set used by
the witness algorithm, then the witness algorithm is $k$-update competitive.
\end{theorem}

Theorem~\ref{th:global_bound}  was proved in a slightly different 
setting in~\cite{Bruce_et_al/05},
but the proof carries over to the present setting
in a straightforward way by using Lemma~\ref{lem:narrow}.

\section{Minimum Spanning Trees with Edge Uncertainty}
\label{sec:edgeu}%
In this section we 
present
an algorithm \code{\tempone} for the problem \code{mst-edge-uncertainty}. In
the case that all areas of uncertainty are either open or trivial,
algorithm \code{\tempone} is $2$-update competitive, which we show is optimal. 
Furthermore,
we show that for arbitrary areas of uncertainty there is no constant
update competitive algorithm.

First, let us recall a well known property, usually referred to
as the \emph{red rule} \cite{Tarjan/83b}, of MSTs:

\begin{proposition}
  \label{out}
  Let $G$ be a weighted graph and let $C$ be a cycle in~$G$. 
  If there exists an edge $e \in C$ with 
  $w_e > w_{e'}$ for all $e' \in C -\{e\}$,
  then $e$ is not in any MST of~$G$.
\end{proposition}

We will use the following notations and definitions:
A graph $\U=(V,E)$ with an area $A_e$ for each edge $e\in E$
is called an {\it edge-uncertainty graph}.
We say a weighted graph
$G=(V,E)$ with edge weights $w_e$ is a {\it realization} of $\U$ if $w_e
\in A_e$ for every $e \in E$.
Note that $w_e$ is associated with $G$ and $A_e$ with $\U$.
We also say that an edge $e$ is {\it trivial} if the area $A_e$ is trivial.

For an edge $e$ in an edge-uncertainty graph we denote the upper limit
of $A_e$ 
by $U_e = \sup A_e$ and the lower limit
of $A_e$ by $L_e = \inf A_e$.

We extend the notion of an MST to edge-uncertainty
graphs in the following way:
Let $\U$ be an edge-uncertainty graph. We say $T$ is an
\emph{MST of $\U$} if $T$ is an MST
of every realization of~$\U$. 
Clearly not every edge-uncertainty graph has an MST.

Let $C$ be a cycle in $\U$. We say the edge $e \in C$ is an
{\it always maximal} edge in $C$ if $L_e \ge U_c$ for all $c \in C-\{e
\}$. Therefore in every realization $G$ of $\U$ we have $w_e \ge w_c$ for
all $c \in C-\{e\}$.

Note that a cycle can have more than one always maximal edge
and not every cycle has an always maximal edge. The following
lemma deals with cycles of the latter kind:

\begin{lemma}
  \label{fg}
  Let $\U$ be an edge-uncertainty graph. Let $C$ be a cycle in $\U$.
  Let $C$ not have an always maximal edge. Then for any $f \in C$ with 
  $U_f = \max\{U_c \mid c\in C\}$ we have that $f$ is non-trivial and
  there exists an edge $g \in C-\{f\}$ with $U_g > L_f$.
\end{lemma}

\begin{proof}
  Let $f \in C$ be an edge with $U_f = \max\{U_c \mid c \in C\}$. 
  If $L_f=U_f$ the edge $f$ would be always maximal. Hence $L_f$ must
  be strictly smaller than $U_f$ and $f$ is non-trivial. 
  Since there is no always maximal edge in $C$,  we
  have that $L_f < \max\{U_c \mid c \in C-\{f\}\}$. 
  Therefore there exists at least one edge $g$ in $C-\{f\}$ with
  $L_f < U_g$.
\end{proof}

\begin{proposition}
  \label{path1}
  Let $\U$ be an edge-uncertainty graph with an MST
  $T$. Let $f=\{u,v\}$ be an edge of $ \U$ such that $f \not \in T$.
  Let $P$ be the path in $T$ connecting $u$ and $v$, then $U_p \le L_f$
  for all $p \in P$.
\end{proposition}
 
\begin{proof}
  Assume there exists a $p \in P$ with $U_p > L_f$. Then there exists
  a realization $G$ of $\U$
  with $w_p > w_f$. Hence by removing the edge $p$ and adding the edge
  $f$ to $T$ we obtain a spanning tree that is cheaper than
  $T$. So $T$ is not an
  MST for $G$. This is a contradiction since $T$ is
  an MST of $\U$ and therefore of any realization
  of~$\U$.
\end{proof}

Our algorithm \code{\tempone} applies
the red rule to the given uncertainty graph, but we
have to be careful about the order in which edges are considered.
The order we use is as follows:
Let $\U$ be an edge-uncertainty graph and let $e,f$ be
two edges in $\U$. We say 
\begin{itemize}
\item[] $e < f$ if $L_e<L_f$ or ($L_e = L_f$ and $U_e < U_f$),
\item[] $e \le f$ if $e < f$ or ($L_e=L_f$ and $U_e=U_f$). 
\end{itemize}
Edges with the same upper and lower weight limit are ordered
arbitrarily.

\begin{figure}
\hspace*{0.15\textwidth}\begin{minipage}{0.7\textwidth}
01 Index all edges such that 
$e_1 \le e_2 \le \dots \le e_m$.\\
02 Let $\G$ be $\U$ without any edge\\
03 \textbf{for} $i$ from $1$ to $m$ \textbf{do}\\
04 \hspace{1cm} add $e_i$ to $\G$ \\
05 \hspace{1cm} \textbf{if} $\G$ has a cycle $C$ \textbf{then}\\
06\hspace{2cm} \textbf{if} $C$ contains an always maximal edge $e$ \textbf{then}\\
07\hspace{3cm} delete $e$ from $\G$\\
08\hspace{2cm}  \textbf{else} \\
09\hspace{3cm} let $f \in C$ such that $U_f = \max\{U_c| c\in C\}$ \\
10\hspace{3cm} let $g \in C-\{f\}$ such that $U_g > L_f$ \\
11\hspace{3cm} update $f$ and $g$ \\
12\hspace{3cm} restart the algorithm \\
13\hspace{2cm} \textbf{end if} \\
14\hspace{1cm} \textbf{end if}\\
15  \textbf{end for} \\
16  \textbf{return} $\G$
\end{minipage}
\caption{Algorithm \code{\tempone}}
\label{fig:algo}
\end{figure}%
Algorithm \code{\tempone} is shown in Figure~\ref{fig:algo}. Observe that:
\begin{itemize}
\item In case no update is made the algorithm \code{\tempone} will
  perform essentially Kruskal's algorithm \cite{Cormen_et_al/01}.
  When a cycle is created
  there will be an always maximal edge in that cycle.
  Due to the order
  in which the algorithm adds the edges to $\G$ the edge
  $e_i$ that closes a cycle $C$ must be an always maximal edge in
  $C$. So where Kruskal's algorithm
  does not add an edge to $\G$ when it would close a cycle, the \code{\tempone}
  algorithm adds this edge to $\G$ but then deletes it or an
  equally weighted edge in the cycle from $\G$. 
\item
  By Lemma \ref{fg} the edges $f,g$ in line $9$ and $10$ exist and $f$
  is non-trivial.
\item The algorithm will terminate.
  The algorithm either updates at
  least one non-trivial edge $f$ and restarts, or does not perform any
  updates. Hence the algorithm \code{\tempone} will eventually return an
  MST of $G$.
\item During the run of the algorithm the graph $\G$ is either a
  forest or contains one cycle. In case the most recently added edge
  closes a cycle either one edge of the cycle will be deleted or after
  some updates the algorithm restarts and $\G$ has no edges. Hence at
  any given time there is at most one cycle in $\G$.
\end{itemize}

As the algorithm may restart itself, we say a \emph{run} is completed if the
algorithm restarts or returns the MST. In case of a
restart, another run of the algorithm starts.

Before showing that the algorithm \code{\tempone} is  $2$-update
competitive under the restriction to open or trivial areas,
we discuss some technical preliminaries.
In each run the algorithm considers all edges in a certain order
$e_1, \dots, e_m$. 
During the run of the algorithm we refer to the currently considered
edge as $e_i$. Let $u$ and $v$ be two distinct vertices. In case $u$ and $v$
are in the same connected component of the subgraph with edges
$e_1,\ldots,e_{i-1}$, then they are also 
connected in the current $\G$. Furthermore, we need
some properties of a path connecting $u$ and $v$ in $\G$ under certain
conditions. 
The next two lemmas establish these properties. They are technical and
are solely needed in the proof of Lemma \ref{witness_mst_edge}.  

\begin{lemma}\label{lem:singleedge}
  Let $h=\{u,v\}$ and $e$ be two edges in $\U$. Let $h\neq e$ and
  $L_h<U_e$. Let the algorithm be in a state such that $h$ has been
  considered. Then $u$ and $v$ are connected in the current $\G -
  \{e\}$.
\end{lemma}

\begin{proof}
  If the edge $h$ is in the current $\G$ then clearly $u$ and $v$ are
  connected in $\G-\{e\}$, so assume that $h$ is no longer in $\G$. 
  Therefore it must have been an always maximal edge in a cycle $C$.
  In order for $h$ to be an always maximal edge in $C$ we must have
  that $L_c \le U_c \le L_h$ for all $c \in C - \{h\}$. So since $L_h
  < U_e$ we have that $L_c < U_e$. Also the edge $h$ can not be
  an always maximal edge in $C$ if $C$ contains $e$.

  Clearly $C-\{h\}$ is a path in $U$ connecting $u$ and $v$ and does
  not contain $e$. Since the edges in $C -\{h\}$ might have been
  deleted from the current $\G$ themselves we have to use
  this argument repeatedly, but eventually we get a path in the
  current $\G-\{e\}$ connecting $u$ and $v$.
\end{proof}

The next lemma follows directly from Lemma~\ref{lem:singleedge}.

\begin{lemma}
  \label{path}
  Let $u,v$ be vertices and $e$ be an edge in $\U$.
  Let $P$ be a path in $\U-\{e\}$ connecting $u$ and $v$ with
  $L_p<U_e$ for all $p \in P$.
  Let the algorithm be in a state such that all edges of $P$ have been
  considered,
  then there exists a path $ P'$ in the current $\G$ connecting $u$ and
  $v$ with $e \not \in P'$.
\end{lemma}

\begin{lemma}
  \label{witness_mst_edge}
  Assume that all uncertainty areas are open or trivial.
  The edges $f$ and $g$ as described in the algorithm  \code{\tempone} at
  line $9$ and $10$ form a witness set.
\end{lemma} 

\begin{proof}
We have the following situation: There exist a cycle $C$ in $\G$ 
with no always maximal edge. Let $m =  \max \{U_c~|~ c \in C\}$.
The edges $f$ and $g$ are in $C$ with $U_f = m$
and $U_g>L_f$.
By Lemma \ref{fg}
the area $A_f$ is non-trivial.

We now assume that the set $\{f,g\}$ is not a witness set. 
So we can update some edges, but not $f$ or $g$ such that
the resulting edge-uncertainty graph $\U'$ has an MST
$T$. Let $U'_e$ and $L'_e$ denote the upper and lower limit of an area
for an edge $e$ with regard to $\U'$. Since both edges $f$ and $g$
are not updated we note that 
$$ L_f=L'_f,  U_f=U'_f, L_g=L'_g, U_g=U'_g.$$

Since all areas in $\U'$ and $\U$ are either trivial or open, and $C$
has no always maximal edge, the weight of every edge in $C$ must be
less than $m$. In particular we have that for all $c \in C$
 $$U'_c <m  \mbox{ or } L'_c < U_c=m.$$
Since $U_f=m$ there exists a realization $G'$ of  $\U'$ and $\U$, 
where the weight of $f$ is greater than the weight of any other edge
in $C$. 
By Proposition \ref{out} the edge $f$ is not in any MST of $G'$ and therefore also not in $T$.

Let $u$ and $v$ be the vertices of $f$. By Proposition \ref{path1}
there exists a path $P$ in $\U'$ connecting $u$ and $v$ with $U'_p \le
L_f$ for all $p\in P$. Since $U_g > L_f$ and neither $f$ nor $g$ are
updated the edge $g$ is not in the path $P$.
We now argue that all edges of $P$ must have been already considered
by the algorithm. For this we look at the following two cases:
 
\noindent Case 1) Let $p \in P$ and $L'_p < L_f$. Since $L_p \le L'_p$
we have that $L_p < L_f$.

\smallskip
\noindent Case 2) Let $p \in P$ and $L'_p=L_f$. Since $U'_p \le L_f$ we have that $L'_p=U'_p=L_f$.
Either the area $A_p$ is also trivial
($L_p=U_p=L'_p=U'_p=L_f$) or $A_p$ is open and contains the point $L'_p$, in
this case $L_p<L'_p$.

\smallskip

So for all $p\in P$ we have $ L_p <L_f \mbox{ or } L_p=U_p=L_f<U_f$.
Therefore  all edges of $P$ will be considered before $f$.
We also note that $L_p \le L'_p \le L_f < U_g$ for all $p \in P$. 
By Lemma \ref{path} there exists a path $P'$ in $\G$ connecting $u$ and $v$
and $g \not \in P'$. Hence $\G$ has two cycles, which is a contradiction.
\end{proof}

Using Theorem~\ref{th:global_bound},
this leads directly to the following result.

\begin{theorem}
\label{2competitive}
Under the restriction to open and trivial areas the algorithm
\code{\tempone} is $2$-update competitive.
\end{theorem}

We remark that the analysis of algorithm \code{\tempone}
actually works also in the more general setting
where it is only required that every area is
trivial or satisfies the following condition:
the area contains neither its infimum nor its
supremum.
It remains to show that under the restriction to open and trivial
areas there is no algorithm for the \code{mst-edge-uncertainty}
problem that is $(2-\epsilon)$-update competitive.

\begin{figure}[!ht] 
  \center{
    \scalebox{.30}{\includegraphics{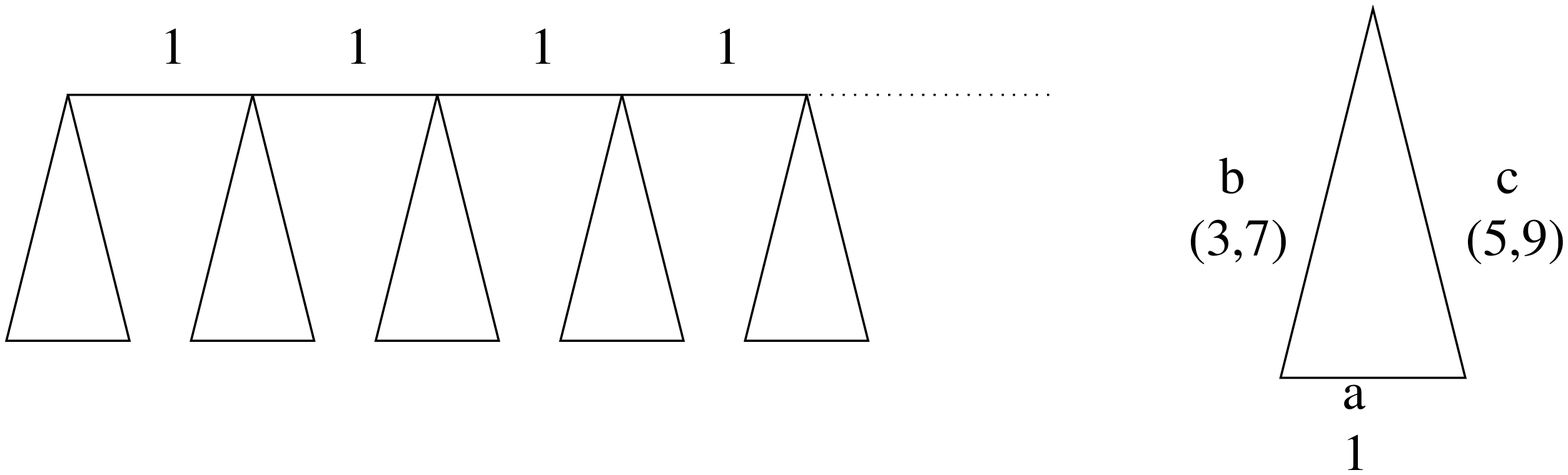}}
  }
  \caption{Lower bound construction}
  \label{lowerbound}
\end{figure}

\begin{example}
The graph $G$ displayed in Figure \ref{lowerbound} consists of
a path and, for each vertex of the path, a gadget connected to
that vertex.
Each gadget is a triangle with sides $a,b$ and $c$ and areas
$A_a=\{1\}$, $A_b= (3,7)$ and $A_c= (5,9)$. In each gadget  $a$
and either $b$ or $c$ are part of the minimum spanning tree.
If the algorithm updates $b$ we let the weight of $b$ be $6$. So
$c$ needs to be updated, which reveals a weight for $c$ of $8$. However,
by updating only $c$ the edge $b$ would be part of the minimum spanning
tree regardless of its exact weight.  If the algorithm
updates $c$ first, we let the weight of $c$ be $6$. The necessary update of
$b$ reveals a weight of $4$, and updating only $b$ would have been enough. So in
each gadget every algorithm makes two updates where only one is
needed by $\OPT$. Hence no deterministic algorithm is $(2-\epsilon)$-update competitive.
\end{example}

The following example shows that without restrictions on the areas there is no
algorithm for the \code{mst-edge-uncertainty} problem that is
constant update competitive. 

\begin{example}\label{ex:ne}
Figure \ref{badmst-both}(a) shows an example of an edge-uncertainty graph
for which no algorithm can be constant update
competitive. 
The minimum spanning tree consists of all edges incident with $u$ and all
edges incident with $v$ plus one more edge. Let us assume the weight of one of the
remaining $k=(n-2)/2$ edges is $2$ and the weight of the others is $3$. Any algorithm
would need to update these edges until it finds the edge with weight
$2$.
This in the worst case could be the last edge and $k$ updates were made. 
However $\OPT$ will only
update the edge with weight $2$ and therefore $\OPT=1$.%
\begin{figure}[tbh]
  \centering
    \scalebox{.5}{ \includegraphics{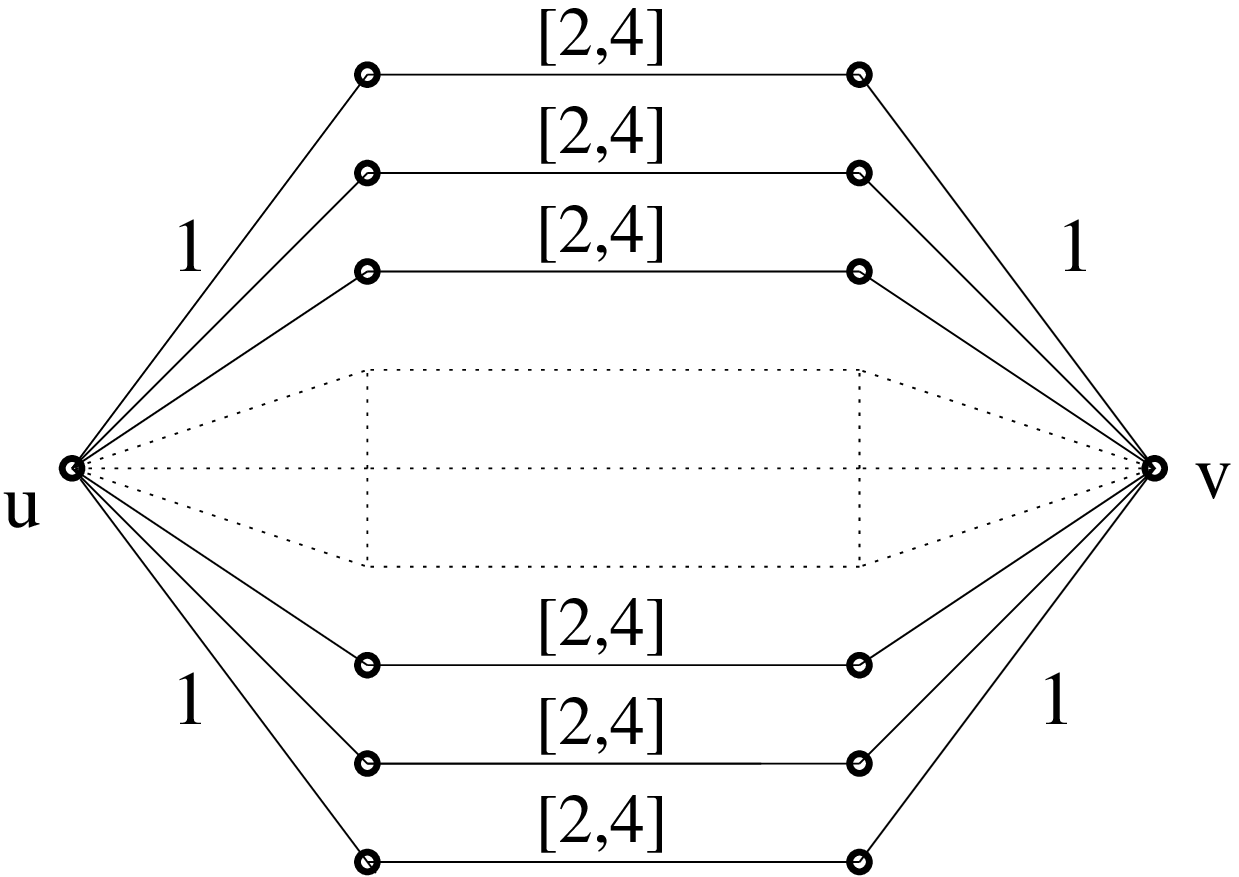}}
    \hspace{3em}
    \scalebox{.27}{ \includegraphics{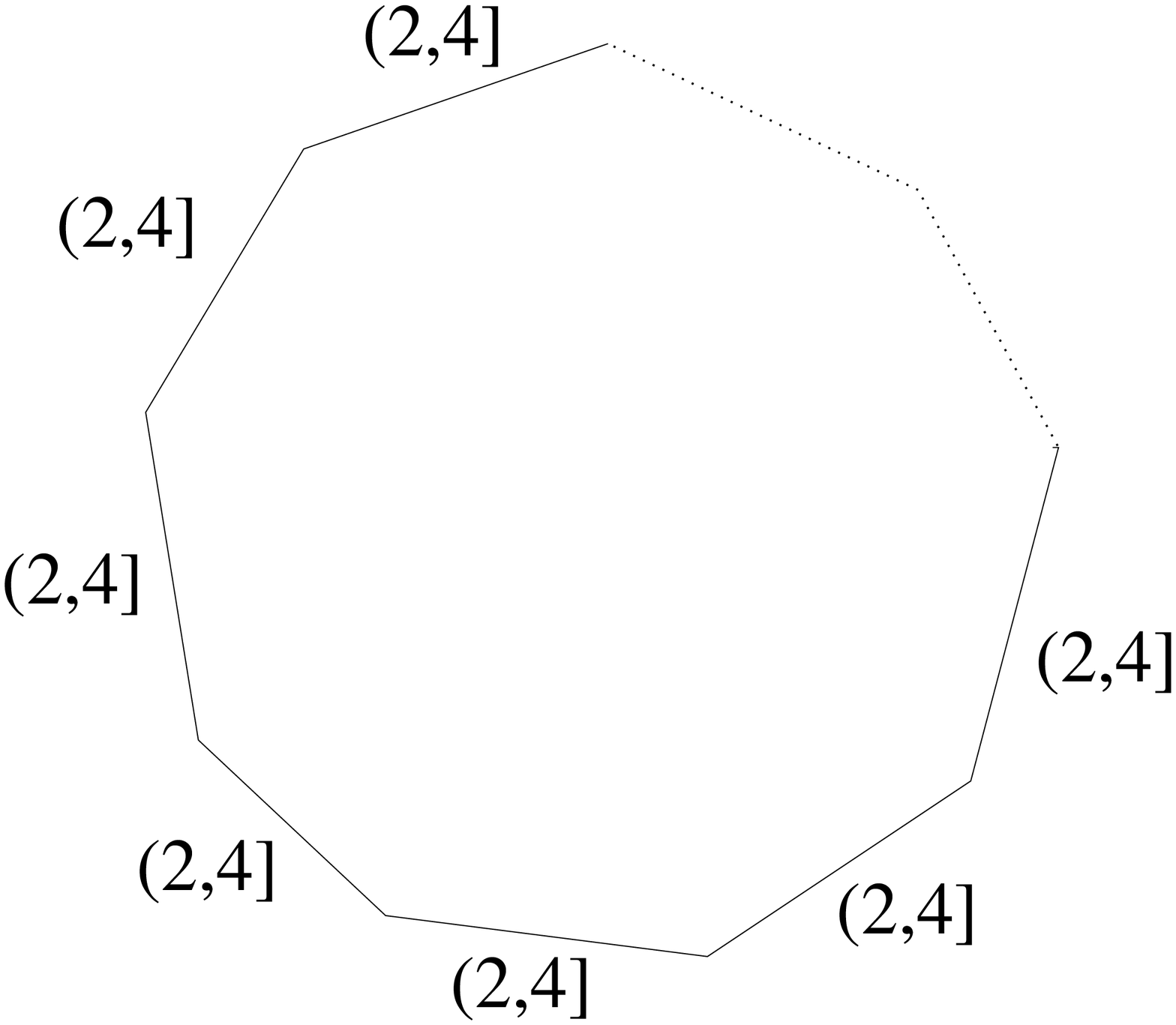}}
    \\
    \hspace*{0.7cm}(a)\hspace{7.3cm}(b)
  \caption{Non-existence of constant update competitive algorithms}
  \label{badmst-both}
\end{figure}

Note that this example actually shows that there is no algorithm that is
better then $(n-2)/2$-update competitive, where $n$ is the number of
vertices of the given graph. By adding edges with uncertainty area
$[2,4]$ such that the neighbors of $u$ and the neighbors
of $v$ form a complete bipartite graph, we even get a lower bound of
$\Omega(n^2)$.
\end{example}

The construction in Example~\ref{ex:ne} works also if the
intervals $[2,4]$ are replaced by half-open intervals
$[2,4)$. Thus, the example demonstrates that with closed lower
limits on the areas there is no constant update competitive
algorithm for the \code{mst-edge-uncertainty} problem.
The following example does the same for closed upper limits.

\begin{example}
The graph shown in Figure~\ref{badmst-both}(b) is one
big cycle with $k$ edges and the uncertainty area of each edge is
$(2,4]$. Let us assume exactly one edge $e$ has weight $4$ and the
others
are of weight $3$. 
In the worst case any algorithm has to update all
$k$ edges before
finding $e$. However $\OPT$ is $1$ by just updating $e$.
\end{example}


\vskip-0.5cm
\section{Minimum Spanning Tree with Vertex Uncertainty}
\label{sec:vertexu}%
In this section we consider the model of vertex-uncertainty
graphs. The models of vertex-uncertainty and edge-uncertainty are
closely related. 
Clearly a vertex uncertainty graph $\U$ has an associated
edge-uncertainty graph $\bar\U$ where the area for each edge
$e=\{u,v\}$ is determined by 
the combinations of possible locations of $u$ and $v$ in $\U$,
i.e., the areas $\bar A$ in $\bar\U$ are defined as
$\bar{A}_{\{u,v\}} = \{d(u',v') | u'\in
A_u, v' \in A_v\}$.

An update of an edge $e=\{u,v\}$ in $\bar\U$ can be performed (simulated)
by updating $u$ and $v$ in $\U$; these two vertex updates might also
reveal additional information about the weights of other edges incident
with $u$ or $v$.
Furthermore, note that if neither of the two vertices $u$ and $v$ in $\U$
is updated, no information about the weight of $e$ can be obtained.
Thus, we get:

\begin{lemma}
Let $\phi$ be a graph problem such that the set of solutions for
a given edge-weighted graph $G=(V,E)$ depends only on
the graph and the edge weights (but not the locations
of the vertices).
Let $\U$ be a vertex-uncertainty graph that is an instance of~$\phi$.
If $\bar W \subseteq E$ is a witness set for $\bar\U$,
then $W= \bigcup_{\{u,v\}\in\bar{W}}\{u,v\}$
is a witness set for $\U$.
\end{lemma}

Using Theorem \ref{th:global_bound} we obtain the following result.

\begin{theorem}
\label{open2k}
Let $\phi$ be a graph problem such that the set of solutions for
a given edge-weighted graph depends only on
the graph and the edge weights (but not the locations
of the vertices). Let $A$ be a $k$-update
competitive algorithm for the problem $\phi$ with respect
to edge-uncertainty graphs. If $A$ is a witness
algorithm, then by simulating an edge update by updating both its endpoints
the algorithm $A$ is $2k$-update competitive for vertex-uncertainty graphs.
\end{theorem}

By standard properties of Euclidean topology,
the following lemma clearly holds.

\begin{lemma}
\label{open2open}
Let $\U$ be a vertex uncertainty graph with only trivial or open
areas. Then $\bar\U$ also has only trivial or open areas.
\end{lemma}

\begin{theorem}
Under the restriction to trivial or open areas the algorithm \code{\tempone}
is $4$-update competitive for the \code{mst-vertex-uncertainty} problem, which
is optimal.
\end{theorem}

\begin{proof}
Combining Theorems \ref{2competitive} and \ref{open2k} together with
Lemma \ref{open2open}, we get that \code{\tempone} is $4$-update
competitive for the \code{mst-vertex-uncertainty} problem when restricted
to trivial or open areas. It remains to show that this is optimal.

\begin{figure}[!ht]
  \centering
    \scalebox{.45}{ \includegraphics{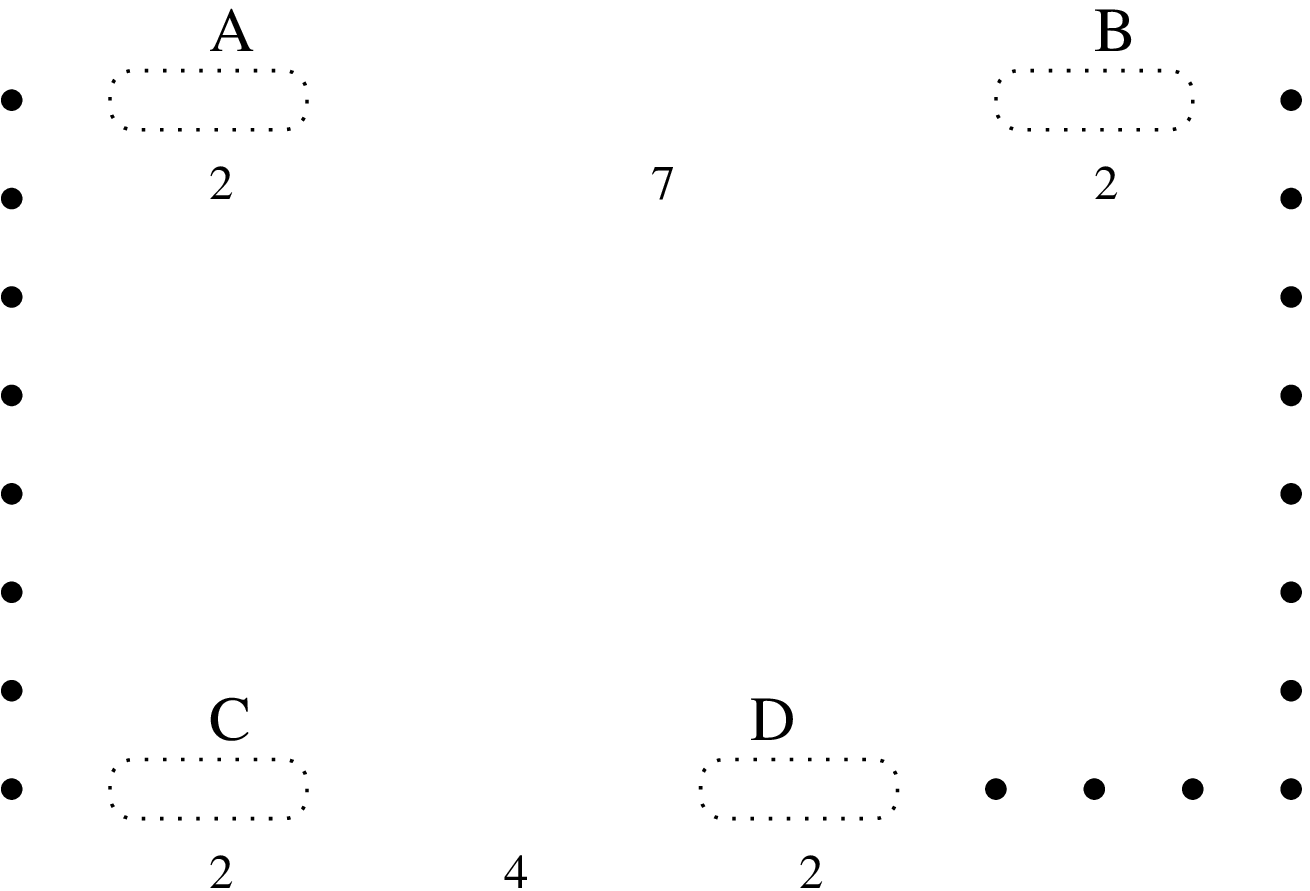}}
    \hspace{4em}
    \scalebox{.45}{ \includegraphics{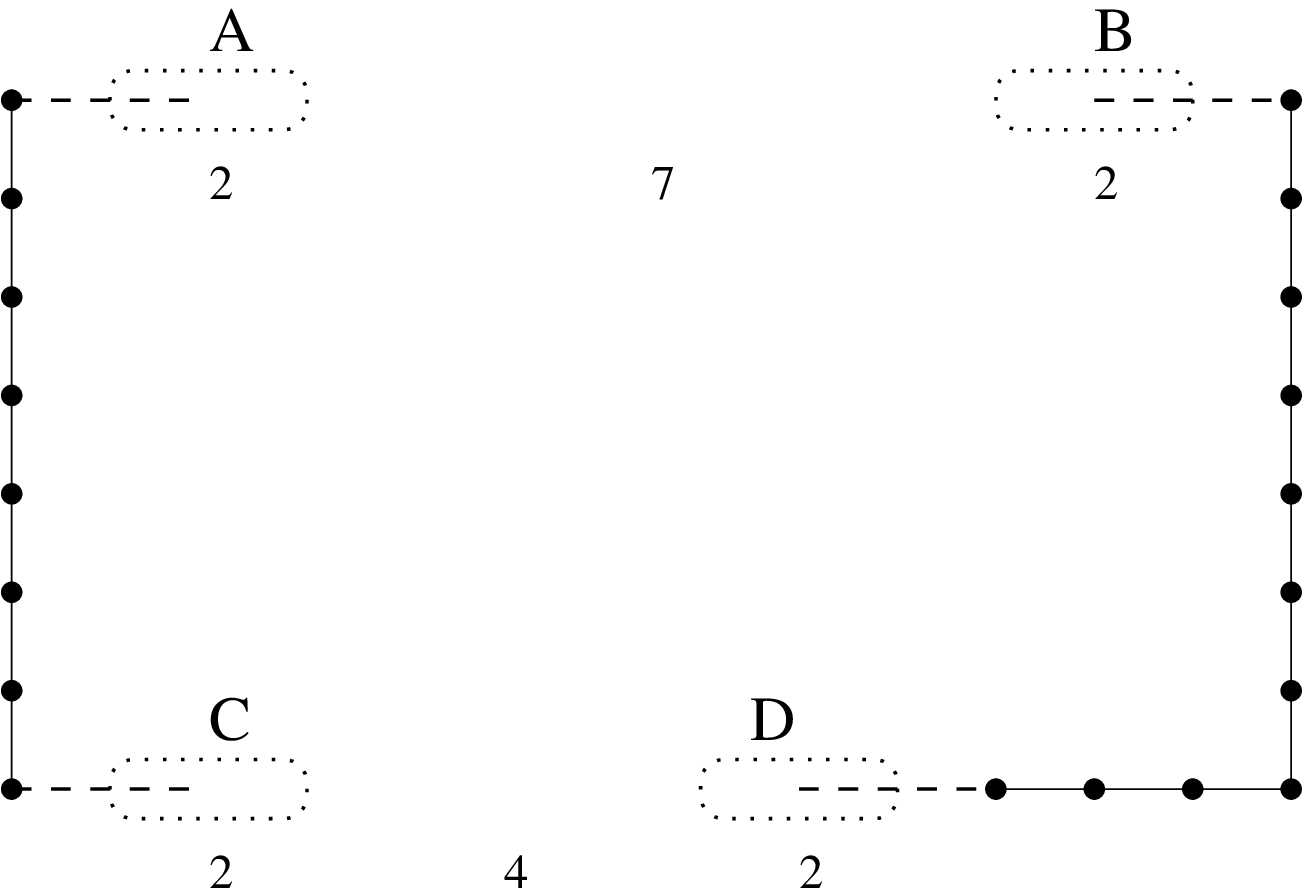}}\\
    (a)\hspace{7cm}(b)
    \vskip-0.2cm
  \caption{(a) Lower bound construction $\;$ (b)
  Edges that are in any minimum spanning tree}
  \label{fig:vertexBound}
\end{figure}%
\vskip-0.1cm
We show that no algorithm can be better than $4$-update
competitive.
In Figure~\ref{fig:vertexBound}(a) we give a construction in the Euclidean
plane for which any algorithm
can be forced to make 4 updates, while $\OPT$ is 1.
The black dots on the left and right represent trivial areas.
The distance between two neighboring trivial areas is~$1$.
There are four non-trivial areas $A, B, C$ and $D$. Each of these
areas is a long, thin open area of length $2$ and small positive
width. The distance between each non-trivial area and its
closest trivial area is~$1$ as well. Let $G$ be the complete graph
with one vertex for each of the trivial and non-trivial
areas.

Independent of the exact locations of the vertices in the non-trivial areas $A, B, C$
and $D$, the edges indicated in Figure~\ref{fig:vertexBound}(b) must be
part of any MST. Note that the distance between the
vertex of a non-trivial area and its trivial neighbor is in $(1,3)$ and
thus less than~$3$.

We now consider the distances between the non-trivial areas. We let
$d(X,Y)$ be the area of all possible distances between two vertex areas
$X$ and $Y$. So $d(A,B) = (7,11)$, $d(C,D) = (4,8)$. Note that the
distance between the vertices in $A$ and $D$ and the distance
between the vertices in $B$ and $C$ are greater than~$8$,
so either the edge $AB$ or the edge $CD$ is part of the minimum
spanning tree.

Every algorithm will update the areas $A, B, C$ and $D$ in a
certain order until it is clear that either the distance
between the vertices of $A$ and $B$ is smaller or equal to the
distance between the vertices of $C$ and $D$, or vice versa.
In order to force the algorithm to update all four areas, we
let the locations of the vertices revealed in any of the first 3
updates made by the algorithm be as follows:
\begin{itemize} 
\item $A$ or $D$: the vertex will be located far to the right,
\item $B$ or $C$: the vertex will be located far to the left.
\end{itemize}
Here, `far to the right' or `far to the left' means that
the location is very close (distance $\varepsilon>0$, for
some small $\varepsilon$) to the right or left end
of the area, respectively.

We show that it is impossible for the algorithm to
output a correct minimum spanning tree after only three
updates. Consider the situation after the algorithm has
updated three of the four non-trivial areas.
Since the choice of the locations of the vertices in the areas is
independent of the sequence of updates, we have to consider four
cases depending on which of the four areas has not yet been updated.
We use $A',B',C'$ and $D'$ to refer to the areas $A,B,C$ and $D$ after
they have been updated.
If the area $A$ is the only area that has not yet been updated, we have that
$d(A,B')=(7+\epsilon,9+\epsilon)$ and $d(C',D')=\{8-2\epsilon\}$. Clearly the area $A$
needs to be updated. By having the vertex of area $A$ on the far left,
updating only area $A$ instead of the areas $B,C,D$ results in
$d(A',B)=(9-\epsilon,11-\epsilon)$ and $d(C,D)=(4,8)$. Hence $\OPT$ would
only update the area $A$ and know that the edge $AB$ is not part of the
minimum spanning tree.
The other three cases are similar.
So for the construction in Figure~\ref{fig:vertexBound}(a), no algorithm can guarantee
to make less than $4$ updates even though a single update is enough for
the optimum. Furthermore, we can create $k$ disjoint copies of the construction
and connect them using lines of trivial areas spaced $1$ apart.
As long as the copies are sufficiently far
apart, they will not interfere with each other. Hence, for a graph with $k$
copies there is no algorithm that can guarantee less than $4k$ updates
when at the same time $\OPT=k$.
\end{proof}

\end{document}